\documentclass{article}
\usepackage{spconf}

\usepackage{amsmath, bbm}
\usepackage{amssymb}
\usepackage{amsfonts}
\usepackage{amsthm}
\usepackage{algorithm}
\usepackage{algorithmic}
\usepackage{nicefrac}
\usepackage{bm}
\usepackage{hyperref}
\usepackage[utf8]{inputenc} 
\usepackage[T1]{fontenc}
\usepackage{graphicx}
\usepackage[font=small]{caption}
\usepackage{tikz}
\usetikzlibrary{positioning}
\usepackage{pgfplots}
\usepackage{epstopdf}
\usepackage[labelformat=simple]{subcaption}
\usepackage{multirow}
\usepackage{lipsum}
\usepackage{mathtools}

\long\def\comment#1{}

\newtheorem{theorem}{Theorem}

\input{mysymbol.sty}

\usepackage{etoolbox} 

\pgfplotsset{compat=1.18}

\begin{document}

\title{
LINK-SHARING BACKPRESSURE ROUTING IN WIRELESS MULTI-HOP NETWORKS
\thanks{
    Emails: $^{\star}$\{zhongyuan.zhao, yujun.ming, segarra\}@rice.edu, $^{\ddag}$\{ananthram.swami, kevin.s.chan, fikadu.t.dagefu\}.civ@army.mil}
}

\name{Zhongyuan Zhao$^{\star}$, Yujun Ming$^{\star}$, 
{Ananthram Swami}$^{\ddag}$, Kevin Chan$^{\ddag}$, Fikadu Dagefu$^{\ddag}$,
{Santiago Segarra}$^{\star}$}
\address{\textit{$^\star$Rice University, USA \hspace{10mm}  \hspace{2mm}  $^\ddag$US Army DEVCOM Army Research Laboratory, USA}
}

\maketitle

\begin{abstract}
Backpressure (BP) routing and scheduling is an established resource allocation method for wireless multi-hop networks, noted for its fully distributed operation and maximum queue stability. 
Recent advances in shortest path-biased BP routing (SP-BP) mitigate shortcomings such as slow startup and random walks, yet exclusive link-level commodity selection still causes last-packet problem and bandwidth underutilization. 
By revisiting the Lyapunov drift theory underlying BP, we show that the legacy exclusive commodity selection is unnecessary, and propose a Maximum Utility (MaxU) link-sharing method to expand its performance envelope without increasing
control message overhead. 
Numerical results show that MaxU SP-BP substantially mitigates the last-packet problem and slightly expands the network capacity region.
\end{abstract}

%

\section{Introduction}\label{sec:intro}

Backpressure (BP) routing \cite{tassiulas1992, neely2005dynamic, georgiadis2006resource} is a well-established resource allocation method for wireless multi-hop networks, applicable to mobile ad-hoc/sensor networks, xG (device-to-device, wireless backhaul, and non-terrestrial coverage), vehicular, Internet-of-Things (IoT), and machine-to-machine (M2M) communications \cite{kott2016internet,akyildiz20206g,chen2021massive,noor20226g}.
BP schemes rely on per-destination queues, where packets destined to node $c$ are referred as commodity $c$.
As a framework for joint routing and scheduling built on per-link exclusive commodity selection and MaxWeight scheduling \cite{tassiulas1992}, BP effectively manages inter-flow interference and enables packets to explore all possible routes. 
Theoretically, it guarantees maximum queue stability within the network capacity region (throughput optimality)~\cite{neely2005dynamic, georgiadis2006resource, neely2022stochastic}. 
In practice, BP supports fully distributed operations using distributed schedulers \cite{joo2012local,zhao2021icassp,zhao2022twc,zhao2023graphbased}, offering efficiency, scalability, and robustness against single-point-of-failure, which are critical for military, disaster relief, and xG backhaul communications.
BP algorithms also extend to edge computing \cite{Kamran2022deco,lin2020distributed,zhao2025icassp} and traffic signal control \cite{levin2023maxpressure}.

The original BP scheme has well-known shortcomings such as slow startup, random walk, and the last-packet problem (LPP)~\cite{neely2005dynamic,georgiadis2006resource,neely2022stochastic,jiao2015virtual,cui2016enhancing,gao2018bias}. 
To address these, various improvements have been proposed, including queue-agnostic biases~\cite{neely2005dynamic, georgiadis2006resource, neely2022stochastic, jiao2015virtual, zhao2023icassp, zhao2023enhanced, zhao2024tmlcn}, virtual queues~\cite{ji2012delay, cui2016enhancing, hai2018delay, zhao2024tmlcn}, and route restrictions~\cite{ying2010combining, Rai2017loop, yin2017improving}. 
SP-BP generalizes BP by adding a distance potential to the backlog while retaining its throughput optimality~\cite{neely2005dynamic, georgiadis2006resource}. 
The state-of-the-art (SOTA) SP-BP~\cite{zhao2023icassp, zhao2023enhanced, zhao2024tmlcn} resolves slow startup and random walk, improving latency and throughput with minimal additional overheads. 
In particular, the shortest path bias is computed at near linear time complexity~\cite{bernstein2019distributed} based on optimally scaled, delay-aware edge weights that encode long-term link rates and conflict relationships of wireless links~\cite{zhao2023icassp, zhao2023enhanced, zhao2024tmlcn}. 

However, the last-packet problem persists due to exclusive commodity selection, which can starve short-lived traffic lacking sustained pressure from newly injected packets~\cite{Alresaini2016bp,ji2012delay,Erfaniantaghvayi2024}.
Existing solutions include virtual queues (HOL~\cite{ji2012delay}, SJB~\cite{Alresaini2016bp}, expQ~\cite{zhao2024tmlcn}) that prioritize older packets and separated queueing systems~\cite{Erfaniantaghvayi2024}, but these approaches can shrink the network capacity region~\cite{zhao2024tmlcn,Erfaniantaghvayi2024}.
Moreover, with increasingly diverse applications such as M2M~\cite{chen2021massive} and IoT~\cite{kott2016internet} characterized by numerous bursty flows, exclusive commodity selection can significantly underutilize advanced high-bandwidth air-interfaces.
This issue cannot be fully addressed by employing larger packet sizes or shorter time slots, due to inherent limits in channel coding, synchronization, and transmit/receive (T/R) switching in radio frontends.

To address this fundamental limitation of SP-BP, we revisit its underlying Lyapunov drift theory~\cite{neely2005dynamic,georgiadis2006resource,neely2022stochastic} and show that link-level exclusive commodity selection is unnecessary.
We introduce a link-sharing approach of MaxU multi-commodity selection, integrate it into SOTA SP-BP~\cite{zhao2024tmlcn} to form MaxU SP-BP, substantially mitigating the last-packet problem while improving link bandwidth utilization.

\noindent
{\bf Contribution:} The contributions of this paper are as follows:\\
\textit{1) Algorithmic:} We develop MaxU multi-commodity selection to replace the  exclusive approach in the BP family, improving latency and link utilization under diverse traffic.\\
\textit{2) Theoretical:} We prove that MaxU SP-BP does not shrink the network capacity region of classic SP-BP, thereby preserving throughput optimality (or maximum queue stability). \\
\textit{3) Empirical:} Our simulation shows that MaxU SP-BP expands the performance envelope of classic SP-BP by significantly reducing latency under mixed streaming and bursty traffic and slightly increasing end-to-end throughput.

\noindent
{\bf Notation:} 
$ |\cdot| $ represents the cardinality of a set.
$ \mathbbm{1}(\cdot) $ is the indicator function.
$ \mathbb{E}(\cdot) $ stands for expectation.
Upright bold lower-case symbol, e.g., $\bbz$, denotes a column vector, and $\bbz_i$ denotes the $i$-th element of vector $\bbz$. 
Upright bold upper-case symbol, e.g., $\bbZ$, denotes a matrix, and calligraphic upper-case symbol denotes a set, e.g., $\ccalG$ for a graph.

\section{Problem Formulation}
\label{sec:sys}

We model a wireless multi-hop network with a  \emph{connectivity graph} $\ccalG^n=(\ccalV,\ccalE)$, and a \emph{conflict graph} $\ccalG^c=(\ccalE,\ccalH)$, as illustrated in Fig.~\ref{F:graph}.
In the connectivity graph $\ccalG^n$, a node $i\in\ccalV$ denotes a device $i$, and a directed edge $(i,j)\in\ccalE$ represents that node $i$ can transmit data to node $j$ directly over-the-air.
Graph $\ccalG^{n}$ is assumed to be strongly connected, meaning that a directed path always exists between any pair of nodes.

In the conflict graph $\ccalG^c$, each vertex $e\in\ccalE$ corresponds to a link in $\ccalG^{n}$ (hence the notation $\ccalE$ is reused) and an undirected edge $(e_1, e_2)\in\ccalH$ captures a conflict between wireless links $e_1, e_2\in\ccalE$ in $\ccalG^{n}$.
Two wireless links may be in conflict due to 1) \emph{interface conflict} (sharing the same transceiver) or 2) \emph{wireless interference} (simultaneous transmission causing outage).
In our distributed approach, the global knowledge of $\ccalG^c$ is not required, instead, each transceiver only needs to maintain a list of conflict neighbors by channel monitoring~\cite[Supp]{zhao2022twc}.

We consider a time-slotted medium access control (MAC) in the wireless network.
Vector $\grave{\bbr}(t)=[\grave{\bbr}_{ij}(t)\mid(i,j)\in\ccalE]\in\reals^{|\ccalE|}$ collects the real-time link rates 
of all links in time slot $t$.
Vector $\bbr=\mathbbm{E}_{t}[\grave{\bbr}(t)]\in\reals^{|\ccalE|}$ collects the long-term average link rates. 
Link rates are measured in packets per time slot.

\vspace{1mm}
\noindent\textbf{Queueing system:} 
We denote the set of all commodities (destinations) as $\ccalC=\ccalV\cup\ccalW$, where $\ccalW$ is the set of virtual sinks in service-centric networks~\cite{zhao2025icassp}, and $\ccalV\cap\ccalW=\emptyset$. 
Each device hosts per-commodity queues with a first-in-first-out (FIFO) discipline.
The queue length of commodity $c$ on device $i\in\ccalV$ at the beginning of time slot $t$ is denoted by $Q_i^{(c)}(t)$.
We have $ Q_i^{(i)}(t)=0 $ and  $Q_i^{(c)}(t)$ for $c\neq i$ evolves as follows: 
\begin{subequations}\label{E:queue}
\begin{align}
& Q_{i}^{(c)}\!(t+1) =\! Q_{i}^{(c)}(t) \!-\! M_{i-}^{(c)}(t) + M_{i+}^{(c)}(t) + A_{i}^{(c)}(t),\\
& M_{i-}^{(c)}(t) = \!\!\sum_{(i,j)\in\ccalE}\! \mu_{ij}^{(c)}(t),\;\; M_{i+}^{(c)}(t) = \!\!\sum_{(j,i)\in\ccalE} \!\mu_{ji}^{(c)}(t) \;.
\end{align}
\end{subequations}
In~\eqref{E:queue}, for commodity $c$ in time slot $t$, $ M_{i-}^{(c)}(t), M_{i+}^{(c)}(t) $ are the total outgoing and incoming packets on device $i$, 
$\mu_{ij}^{(c)}(t)$ is the number of packets transmitted over link $(i,j)$, 
and $A_{i}^{(c)}(t)$ stands for the number of new packets generated by users on device $i$. 
The biased backlog for commodity $c$ is defined as
\begin{equation}\label{E:formulation:bb}   
 U_{i}^{(c)}(t) = Q_{i}^{(c)}(t) + B_{i}^{(c)}, \;\forall\; i\in\ccalV,c\in\ccalC\;, 
\end{equation}
where $ 0\!\leq \! B_{i}^{(c)}\!\!<\!\infty $ is a queue-agnostic, non-negative bias, representing the shortest path distance from $i$ to $c$~\cite{neely2005dynamic,zhao2024tmlcn}.
The bias matrix $\bbB=[B_i^{(c)} \mid i\in\ccalV, c\in\ccalC ]\in\reals_+^{|\ccalV|\times|\ccalC|}$ can be found by a SOTA weighted all-pairs shortest-path (APSP) algorithm with time complexity near-linear in $|\ccalV|$~\cite{bernstein2019distributed}.
The (back)pressure of commodity $c$ on link $(i,j)$ is defined as
\begin{equation}\label{E:backpressure}
U_{ij}^{(c)}(t)= U_{i}^{(c)}(t) - U_{j}^{(c)}(t), \;\forall\; (i,j)\in\ccalE, c\in\ccalC\;.    
\end{equation}

\begin{figure}[t]
    \centering
    \includegraphics[width=0.95\linewidth]{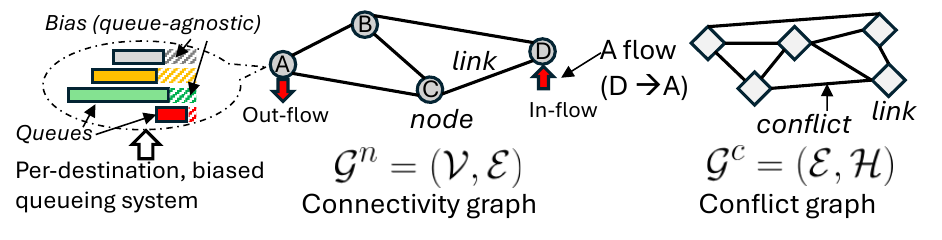}
    \vspace{-0.2in}
    \caption{Graph modeling for wireless multi-hop networks}
    \label{F:graph}
    \vspace{-0.2in}
\end{figure}

\vspace{1mm}
\noindent
\textbf{Lyapunov Drift Minimization:}
SP-BP routing and scheduling is formulated as finding the optimal rate assignments for all link-commodity pairs at every time step $t$, in order to maximize queue stability within the network capacity region.
This is achieved by minimizing the upper bound of the Lyapunov drift of the biased backlogs -- equivalent to solving the following optimization for every time slot $t$~\cite[Theorem 1]{zhao2024tmlcn}:
\begin{subequations}\label{E:formulation}
	\begin{align}
		\bbM^{*}(t) &= \argmax_{\bbM(t)} \sum_{i\in\ccalV}\sum_{j\in\ccalV}\sum_{c\in\ccalC} \mu_{ij}^{(c)}\!(t) U_{ij}^{(c)}\!(t)\label{E:formulation:obj}\\
		\text{s.t. } 
		  \bbM(t) & = \left[\mu_{ij}^{(c)}(t) \mid (i,j)\in\ccalE, c\in\ccalC \right] \;, \label{E:formulation:M}  \\
		 \mu_{ij}^{(c)}\!(t)&\geq 0,\;\forall\; (i,j)\in\ccalE, c\in\ccalC,\; \label{E:formulation:mu} \\
		 \grave{\bbr}_{ij}(t) &\geq  \sum_{c\in\ccalC} \mu_{ij}^{(c)}\!(t), \;\forall\; (i,j)\in\ccalE  \;, \label{E:formulation:link} \\ 
		 Q_{i}^{(c)}(t) & \geq \sum_{j\in\ccalV} \mu_{ij}^{(c)}(t), \;\forall\; i\in\ccalV, c\in\ccalC  \;, \label{E:formulation:comm} \\
      1 & \geq \bbx_{e_1}(t)+\bbx_{e_2}(t),\; \forall\; (e_1,e_2)\in\ccalH\;, \label{E:formulation:conflict}\\
      \bbx_{e}(t)&=\mathbbm{1}\bigg( \sum_{c\in\ccalC}\mu_{ij}^{(c)}(t) \bigg) \;, e=(i,j)\in\ccalE\;. \label{E:formulation:x}
	\end{align} 
\end{subequations}
In \eqref{E:formulation:obj}, the objective function seeks to maximize the rate for those links and commodities with higher pressure [cf.~\eqref{E:backpressure}], hence the name backpressure routing.
The constraints in \eqref{E:formulation} are explained as follows:
\eqref{E:formulation:M} specifies the decision variables as the \emph{final} rate assignments for all link-commodity pairs, e.g.,~\eqref{E:quota} in Section~\ref{sec:solution}.
\eqref{E:formulation:mu} states that the rate assignment of a link-commodity pair should be non-negative. 
\eqref{E:formulation:link} states that the sum rate of all commodities on a link should not exceed the real-time link rate. 
\eqref{E:formulation:comm} states that the sum rate of a commodity across the outgoing links of a device should not exceed its queue length.
\eqref{E:formulation:conflict} and \eqref{E:formulation:x} prohibit simultaneous transmissions on conflicting links.

\vspace{1mm}\noindent
\textbf{Classic SP-BP} approximately solves the optimization in~\eqref{E:formulation} by replacing the constraints in~\eqref{E:formulation:link} and \eqref{E:formulation:comm} with a stricter one, i.e., only one commodity per link per time slot:
\begin{equation}\label{E:optcmdy}
    \sum_{c\in\ccalC}\mathbbm{1}\Big(\mu_{ij}^{(c)}(t)>0\Big) \leq 1,\;\forall\; (i,j)\in \ccalE\;.
\end{equation}
This leads to the following 4-step operations of classic SP-BP~\cite{neely2005dynamic, georgiadis2006resource, neely2022stochastic, zhao2023icassp, zhao2023enhanced, zhao2024tmlcn}, expressed in a slightly different form for the convenience of proof and comparison.

\noindent \textbf{Step 1}. The optimal commodity $c_{ij}^{*}(t)$ on each {directed} link (${i,j}$) is selected as the one with the maximal pressure, i.e., 
\begin{equation}\label{E:commodity}
    c_{ij}^{*}(t) =\argmax_{c\in\ccalV}\; U_{ij}^{(c)}(t) \;,
\end{equation}

\noindent \textbf{Step 2}. All of the real-time link rate $\grave{\bbr}_{ij}(t)$ of link $(i,j)$ is initially assigned to its optimal commodity $c_{ij}^{*}(t)$,
\begin{equation}\label{E:gamma:original}
    \gamma_{ij}^{(c)}\!(t) \!=\! \begin{cases}
         \min\!\left\{\grave{\bbr}_{ij}(t),Q_{i}^{(c)}\!(t)\right\}, & \text{if } c=c_{ij}^{*}\!(t), U^{(c)}_{ij}\!(t)>0, \\
         0, & \text{otherwise}.
    \end{cases}    
\end{equation}
Steps 3--4 follow~\eqref{E:weight} to~\eqref{E:quota} in Section~\ref{sec:solution}.

\section{Maximum Utility (MaxU) SP-BP}\label{sec:solution}

By laying out the actual constraints of routing and scheduling decisions for throughput optimality in~\eqref{E:formulation}, it is evident that the legacy constraint in~\eqref{E:optcmdy} is overly restrictive.
With looser constraints in~\eqref{E:formulation:link} and \eqref{E:formulation:comm}, we can replace the first two steps of classic SP-BP with link-sharing commodity selection to utilize the residual link rate left by the optimal commodity.
This leads to our 4-step operation of MaxU SP-BP:

\vspace{1mm}
\noindent\textbf{Step 1.}
Filter out commodities $c\in\ccalC$ on each link $(i,j)\in\ccalE$:
\begin{equation}\label{E:mu:select}
\alpha_{ij}^{(c)}(t)= \mathbbm{1}\!\left( U_{ij}^{(c)}\!(t) > 0 \right)\cdot\mathbbm{1}\!\left(Q_i^{(c)}\!(t)>0\right)\;.
\end{equation}

\noindent\textbf{Step 2.}
On each directed link $(i,j)$, first sort the remaining commodities decreasingly by pressure: 
\begin{equation}\label{E:sort}
    \tilde{\bbc} = \underset{c\in\ccalC}{\arg\textup{sort}_{\downarrow}} \left[\alpha_{ij}^{(c)}(t) \cdot U_{ij}^{(c)}(t)\right] \;,  
\end{equation}
where vector $\tilde{\bbc}$ collects all commodities in sorted order, and $\tilde{\bbc}_n$ is its $n$th element.
Next, sequentially allocate the residual link rate ${\acute{\bbr}}_{ij}(t)$ to the remaining commodities in sorted order $n=1,2,\dots,|\ccalC|$:
\begin{subequations}\label{E:allocate}
\begin{align}
    \acute{\bbr}_{ij}(t) &= \max\left\{\grave{\bbr}_{ij}(t) -\!\! \sum_{m=1}^{n-1}\gamma_{ij}^{(\tilde{\bbc}_m)}\!(t), 0\right\} \;,\\
    \gamma_{ij}^{(\tilde{\bbc}_n)}\!(t) &= \min\left\{\alpha_{ij}^{(\tilde{\bbc}_n)}\!(t)\cdot\acute{\bbr}_{ij}(t),\; Q_i^{(\tilde{\bbc}_n)}\!(t)\right\}\;.
\end{align}
\end{subequations}
\noindent \textbf{Step 3}. Compute the link utility vector $\bbw(t)\!=\!\big[w_{ij}(t) \mid (i,j)\in\ccalE\big]$, where the utility of link (${i,j}$) is found as
\begin{equation}\label{E:weight}
    w_{ij}(t) =\sum_{c\in\ccalC} \gamma_{ij}^{(c)}(t)\cdot\max\left\{U_{ij}^{(c)}(t),0 \right\}.
\end{equation}
Next, MaxWeight scheduling \cite{tassiulas1992} finds the schedule $\bbx(t)\in\{0,1\}^{|\ccalE|}$ in order to activate a set of \emph{non-conflicting links} achieving the maximum total utility as follows: 
\begin{subequations}\label{E:scheduling}
\begin{align}
    \bbx (t) &= \argmax_{\grave{\bbx} (t)\in \{0,1\}^{|\ccalE|} } ~ \grave{\bbx}(t)^\top  \bbw(t) \;,\label{E:scheduling:obj} \\
    \text{s.t., } \grave{\bbx}_{e_1}(t) &+ \grave{\bbx}_{e_2}(t)\leq 1,\;\forall\; (e_1,e_2)\in\ccalH\;,\label{E:scheduling:conflict}
\end{align}
\end{subequations}
where the constraint in~\eqref{E:scheduling:conflict} is based on \eqref{E:formulation:conflict} and \eqref{E:formulation:x}, which states that conflicting links should not be scheduled together.
\eqref{E:scheduling} defines an NP-hard maximum weighted independent set (MWIS) problem \cite{joo2010complexity} on the conflict graph, which can be approximated by distributed heuristics in practice, such as the LGS~\cite{joo2012local} used in our test and GCN-LGS~\cite{zhao2021icassp, zhao2022twc, zhao2023graphbased}.

\noindent \textbf{Step 4}. 
Apply the schedule to the initial assignment to obtain the final assignment for both routing and link scheduling
\begin{equation}\label{E:quota}
    \mu_{ij}^{(c)}\!(t) = {\gamma_{ij}^{(c)}\!(t)}\cdot \bbx_{ij} (t) ,\;\forall (i,j)\in\ccalE, c\in\ccalC.
\end{equation}
In Fig.~\ref{fig:linksharing}, we illustrate the exclusive and link-sharing approaches under different link rates with a minimal example.

\begin{figure}[t]
    \centering
    \includegraphics[width=0.92\linewidth]{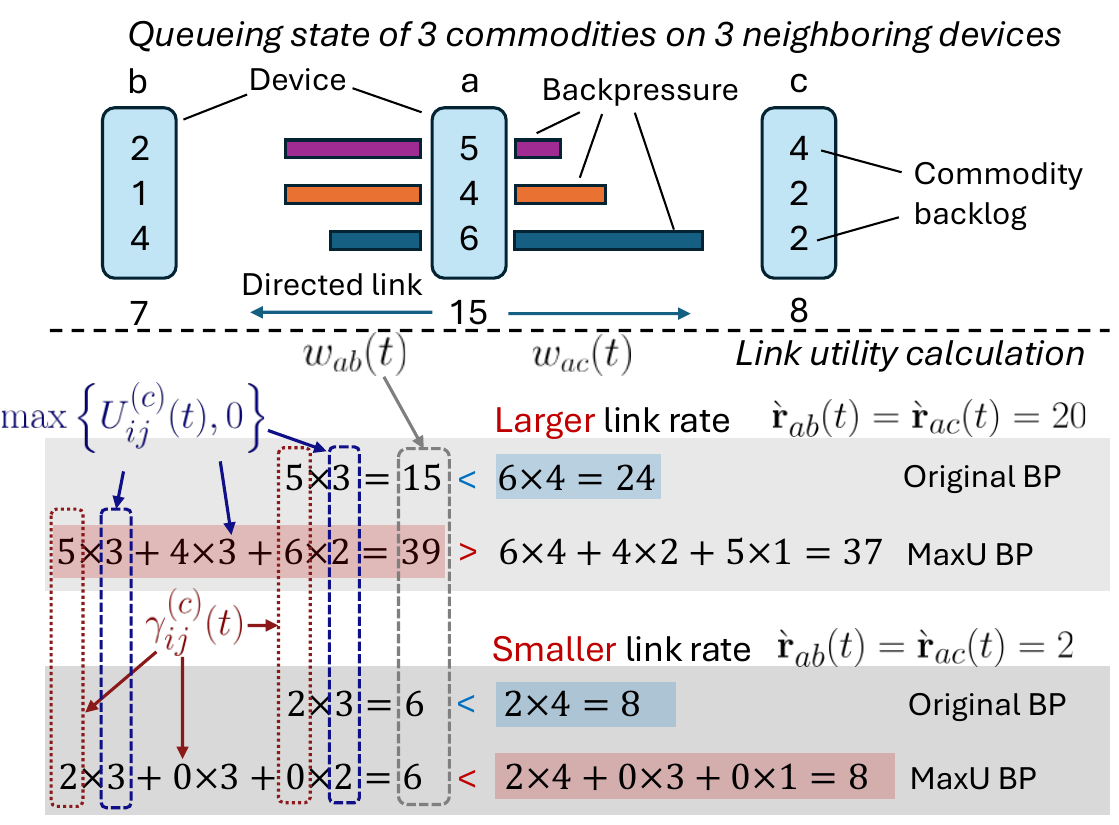}
    \vspace{-0.1in}
    \caption{A mini example of commodity selection and link utility calculation with three devices and three commodities. $B_{i}^{(c)}=0$.}
    \label{fig:linksharing}
    \vspace{-0.2in}
\end{figure}

\vspace{1mm}
\noindent\textbf{Complexity:}
MaxU selection (\eqref{E:mu:select}--\eqref{E:allocate}) has an asymptotic time complexity of $\ccalO(|\ccalC|\log |\ccalC|)$ due to sorting, higher than $\ccalO(|\ccalC|)$ for the legacy approach in~\eqref{E:commodity}--\eqref{E:gamma:original}.

\begin{figure*}[!t]
	\centering
	\vspace{-0.05in}
    \hspace{-3mm}
    \subfloat[]{
    	\includegraphics[height=1.7in]{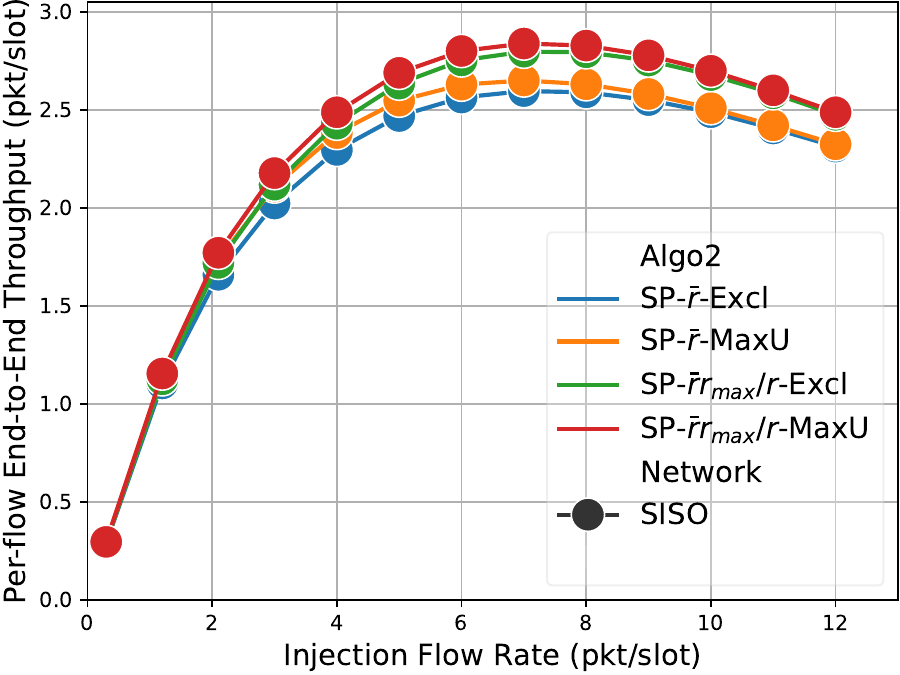}
        \label{fig:thpt}\vspace{-0.1in}
    }
	\subfloat[]{
        \includegraphics[height=1.7in]{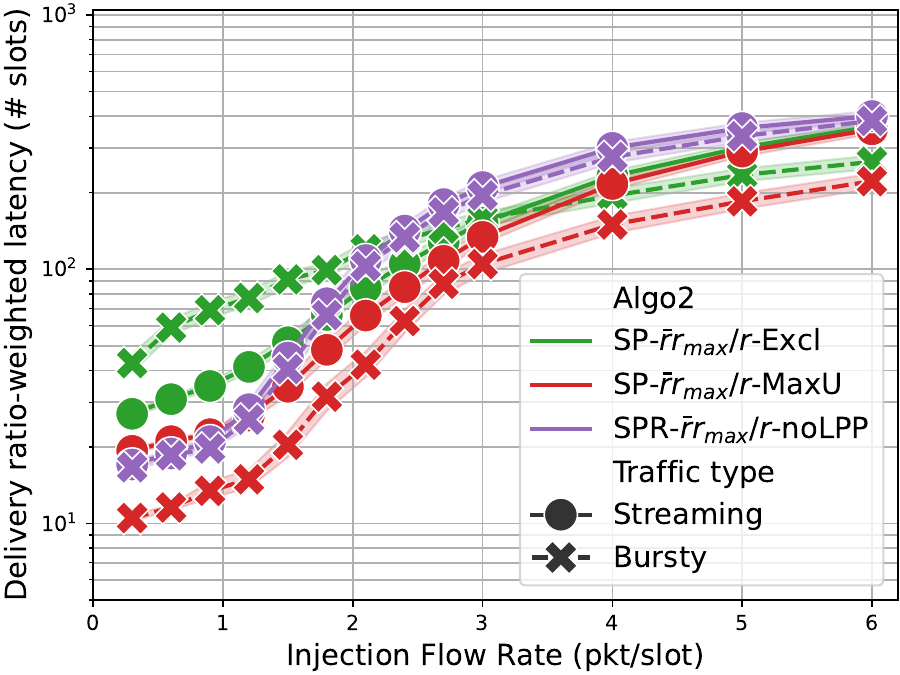}
		\label{fig:e2e_by_lambda}\vspace{-0.1in}
	}
	\subfloat[]{
		\includegraphics[height=1.7in]{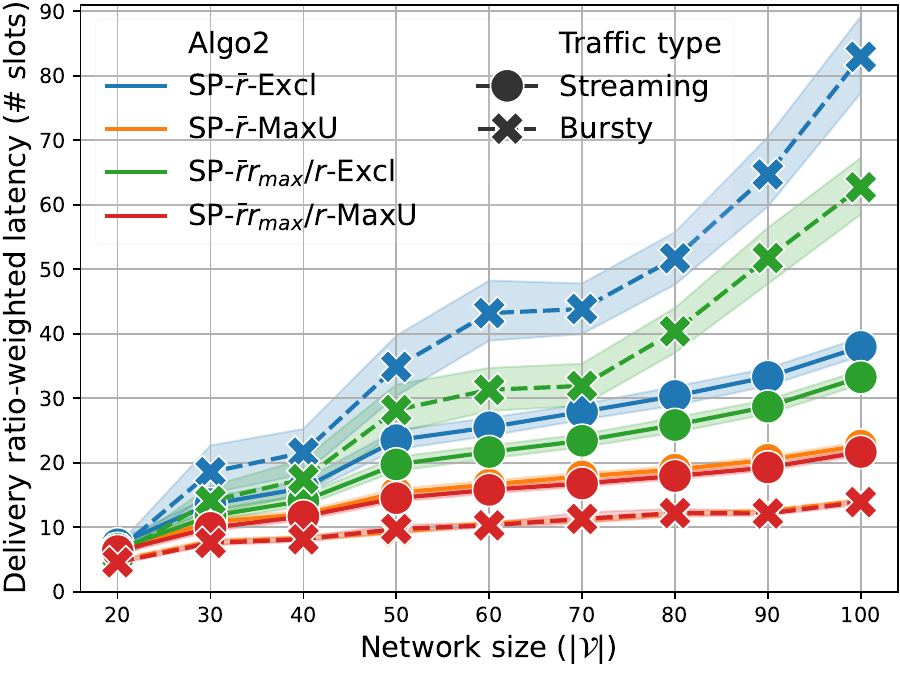}
		\label{fig:network:e2e}\vspace{-0.1in}
	}
	\vspace{-0.1in}
	\caption{
    (a) Average throughput per flow versus flow rate $\lambda$, where all flows are streaming at identical flow rate; 
    (b) Average composite latency (Latency $\times$ delivery ratio  + $T$($1-$ delivery ratio)) by flow rate $\lambda$ under mixed traffic.
    Both (a) and (b) are on random networks of 100 nodes. 
    (c) Average composite latency in networks of 20-100 nodes under mixed traffic.
    $T=1000$.
     The bands indicate $95\%$ confidence interval. 
	}
	\label{fig:network}    
	\vspace{-0.2in}
\end{figure*}

\vspace{1mm}\noindent
\textbf{Throughput Optimality:}
Next, we prove that MaxU SP-BP weakly dominates classic SP-BP in network capacity region, inheriting its proven throughput optimality~\cite[Theorem 1]{zhao2024tmlcn}:
\begin{theorem}\label{th:dominance}
With everything else equal, MaxU SP-BP does not shrink the network capacity region of classic SP-BP.
\end{theorem}
\begin{proof}
We mark the link-commodity rate assignments of the classic SP-BP as $\hat{\bbM}(t)$ and that of MaxU SP-BP as $\tilde{\bbM}(t)$, and likewise for other intermediate variables.
For each link in an arbitrary time slot $t$, there are two mutually exclusive cases:

\noindent
\textit{Case 1:} $\grave{\bbr}_{ij}(t) \leq Q_{i}^{(c_{ij}^{*}(t))}\!(t)$,  steps 1-2 in~\eqref{E:mu:select}--\eqref{E:allocate} are equivalent to steps 1-2 in~\eqref{E:commodity}--\eqref{E:gamma:original}, i.e., $\hat{\gamma}^{(c)}_{ij}(t)=\tilde{\gamma}^{(c)}_{ij}(t)$ for all $(i,j)\in\ccalE, c\in\ccalC$, as illustrated in Fig.~\ref{fig:linksharing}.
In this case, MaxU SP-BP and classic SP-BP are equivalent. 

\noindent
\textit{Case 2:} $\grave{\bbr}_{ij}(t) > Q_{i}^{(c_{ij}^{*}(t))}\!(t)$, the preliminary rate assignment under MaxU SP-BP always includes that of the optimal commodity under the classic SP-BP, given the same queueing state $\bbQ(t)$ and biases $\bbB$.
For all $(i,j)\in\ccalE$, we have 
\begin{equation}\label{E:case2}
\tilde{\gamma}^{(c)}_{ij}(t)\begin{cases}
    = \hat{\gamma}^{(c)}_{ij}(t), & c=c^*_{ij}(t)\\ 
    \geq\hat{\gamma}^{(c)}_{ij}(t)=0, & c\neq c^*_{ij}(t), U^{(c)}_{ij}(t) >0\\ 
    = \hat{\gamma}^{(c)}_{ij}(t) =0, & c\neq c^*_{ij}(t), U^{(c)}_{ij}(t) \leq 0
\end{cases}\;.
\end{equation}

Since case 2 could not be ruled out  when queue stability is maintained under the network capacity region of classic SP-BP (denoted by $\hat\Lambda$),
based on~\eqref{E:case2} and \eqref{E:weight}, we have $\tilde{w}_{ij}(t)\geq\hat{w}_{ij}(t)$, $\forall\; (i,j)\in\ccalE$.
Consequently, with identical $\ccalG^c$, 
\begin{equation}\label{E:mwis:unequal}
\mathbb{E}\Bigg[\sum_{(i,j)\in\ccalE} \tilde{\bbx}_{ij}(t)\tilde{w}_{ij}(t)\Bigg]\geq\mathbb{E}\Bigg[\sum_{(i,j)\in\ccalE} \hat{\bbx}_{ij}(t)\hat{w}_{ij}(t)\Bigg]\;,    
\end{equation} 
since $w_{ij}(t)$ is independent from the topology of $\ccalG^c$.
Based on \eqref{E:weight}, \eqref{E:quota}, and~\eqref{E:mwis:unequal}, the expectation of the objective function in~\eqref{E:formulation} under MaxU SP-BP is greater than or equal to that under the classic SP-BP:
\begin{equation*} 
    \mathbb{E}\Bigg[\sum_{(i,j)\in\ccalE}\sum_{c\in\ccalC} \tilde{\mu}_{ij}^{(c)}\!(t) U_{ij}^{(c)}\!(t)\Bigg] \!\geq\! \mathbb{E}\Bigg[\sum_{(i,j)\in\ccalE}\sum_{c\in\ccalC} \hat{\mu}_{ij}^{(c)}\!(t) U_{ij}^{(c)}\!(t)\Bigg].
\end{equation*}
Based on the above inequality and the proof of throughput optimality in \cite[Theorem 1]{zhao2024tmlcn}, 
for any exogenous packet arrival rate $\lambda\in\hat\Lambda$, 
MaxU SP-BP can also maintain the queueing stability. 
Therefore, $\hat\Lambda\subseteq\tilde\Lambda$, where $\tilde\Lambda$ is the network capacity region of MaxU SP-BP, which corresponds to weak dominance of MaxU SP-BP over classic SP-BP.
\end{proof}

\section{Numerical experiments}
\label{sec:results}
\vspace{-0.05in}
We evaluate four SP-BP variants formed by combining two commodity-selection approaches (\textit{Excl} for classic exclusive selection and \textit{MaxU}) and two shortest path (SP) bias schemes from~\cite{zhao2024tmlcn}:
1) SP-$\bar{r}$, where each edge weight in the connectivity graph $\ccalG^n$ equals the network-wide average long-term link rate $\bar{r}$; and
2) SP-$\bar{r}r_{max}/r$, where the weight of link $e\in\ccalE$ is $\bar{r}r_{max}/r_e$, with $r_{max}$ denoting the maximum long-term link rate across the network\footnote{Code:
\url{https://github.com/zhongyuanzhao/modernSPBP}
}.
We also include a non-BP baseline, SPR-$\bar{r}r_{max}/r$-noLPP, which uses \textit{shortest-path routing (SPR)} with per-neighbor queues to eliminate the last-packet problem.
Simulations are conducted on random networks with settings emulating uniformly distributed wireless devices with omnidirectional antennas and identical transmit power, in fading channels with lognormal shadowing, see the extension of this work in~\cite[Sec. 6]{zhao2025mobihoc} for detailed setups.
Each point in Fig.~\ref{fig:network} is the average of 100 test instances.

In Fig.~\ref{fig:thpt}, we present the average per-flow throughput of the four SP-BP schemes in networks of $100$ nodes, with streaming flows between $40$ distinct source-destination pairs with identical flow rates $\lambda\in[0.1,12]$. 
Consistent with~\cite{zhao2024tmlcn}, incorporating more link features (SP-$\bar{r}r_{max}/r$) increases throughput over SP-$\bar{r}$, which already outperforms many BP variants. 
Importantly, MaxU consistently achieves slightly higher throughput than Excl with all else equal, providing empirical evidence for Theorem~\ref{th:dominance}.

In Fig.~\ref{fig:e2e_by_lambda}, we compare the delivery-ratio-weighted latency of SPR and two SP-BP schemes on networks of 100 nodes under mixed streaming and bursty traffic with identical loads ($\lambda\in\{0.3,\dots,6.0\}$), where bursts start at random times  $\mathbbm{U}(0,T-100)$ and last 30 time slots. 
MaxU SP-BP matches SPR for streaming and outperforms it for bursty traffic under light-to-medium loads ($\lambda\leq1.8$), where classic SP-BP lags, and maintains an advantage over both as load increases.

In Fig.~\ref{fig:network:e2e}, we show the delivery-ratio-weighted latency of the four SP-BP schemes on networks of varying sizes with mixed streaming and bursty traffic under lightweight load ($\lambda_c\sim\mathbb{U}(0.1,1.0)$). 
With all else equal, MaxU cuts the latency of Excl by $37\%\!\sim\!43\%$ for streaming flows and $78\%\!\sim\!84\%$ for bursty flows. 
Under MaxU, bursty traffic even attains lower latency than streaming, whereas under Excl it is the opposite. This result demonstrates that MaxU can effectively prevent starvation of short-lived traffic.

\vspace{-0.1in}
\section{Conclusions}
\label{sec:conclusions}
\vspace{-0.05in}

We introduce link-sharing commodity selection into SP-BP routing and scheduling to improve its link utilization and mitigate the last-packet problem, making it a competitive solution for modern networks with diverse traffic and advanced air-interfaces.
Theoretical and empirical results show that our MaxU SP-BP expands the performance envelope of classic SP-BP with a slightly higher asymptotic time complexity.
Future work will extend this framework to support air-interfaces with MIMO, full-duplex, and multi-modal capabilities.

\pagebreak

\section{Acknowledgment}
Research was sponsored by the US Army DEVCOM ARL Army Research Office and was accomplished under Cooperative Agreement Number W911NF-24-2-0008. 
The views and conclusions contained in this document are those of the authors and should not be interpreted as representing the official policies, either expressed or implied, of the DEVCOM ARL Army Research Office or the U.S. Government. 
The U.S. Government is authorized to reproduce and distribute reprints for Government purposes notwithstanding any copyright notation herein.

{
\footnotesize
\bibliographystyle{ieeetr}
\bibliography{strings,refs,refs_ol}
}

\end{document}